\documentclass{article}

\usepackage{amsmath}
\usepackage{amsthm}
\usepackage{xspace}
\usepackage{color}
\usepackage[colorlinks=true,linkcolor=blue]{hyperref}
\usepackage{amsfonts}
\usepackage{url}
\usepackage{float}
\usepackage{framed}

\usepackage{epsfig}
\usepackage{ifpdf}

\usepackage[a4paper,hmargin=1in,vmargin=1in]{geometry}

\ifpdf

\else

\fi

\newcommand{\set}[1]{\{#1\}}
\newcommand{\Set}[2]{\{#1\,|\,#2\}}

\newcommand{\Q}{\Bbb{Q}}

\newcommand{\eps}{\varepsilon}

\renewcommand{\=}{\!=\!} 
\newcommand{\Com}{\mathsf{Com}}

\newcommand{\Acc}{\mathsf{Acc}}
\newcommand{\Prob}{\mathrm{Prob}}
\newcommand{\dProb}{\mathrm{Prob}^*}

\newcommand{\hP}{P}
\newcommand{\hQ}{Q}
\newcommand{\dP}{\hat{P}}
\newcommand{\dQ}{\hat{Q}}
\newcommand{\etal}{{\it et~al}}

\newtheorem{theorem}{Theorem}[section]
\newtheorem{definition}[theorem]{Definition}
\newtheorem{remark}[theorem]{Remark}
\newtheorem{proposition}[theorem]{Proposition}
\newtheorem{lemma}[theorem]{Lemma}

\allowdisplaybreaks

\title{Multi-Prover Commitments Against Non-Signaling Attacks}

\author{Serge Fehr and Max Fillinger\\[2ex]
\em\normalsize Centrum Wiskunde \& Informatica (CWI), Amsterdam, The Netherlands\\
{\normalsize \tt \{serge.fehr,M.J.Fillinger\}@cwi.nl}}
\date{}

\begin{document}

\maketitle

%=================================================================
%                  ABSTRACT
%=================================================================

\begin{abstract}
We reconsider the concept of two-prover (and more generally: 
multi-prover) commitments, as introduced in the late eighties in the 
seminal work by Ben-Or \etal. As was recently shown by Cr{\'e}peau 
\etal., the security of known two-prover commitment schemes not only 
relies on the explicit assumption that the two provers cannot 
communicate, but also depends on what their information processing 
capabilities are. For instance, there exist schemes that are secure 
against classical provers but insecure if the provers have {\em quantum} 
information processing capabilities, and there are schemes that resist 
such quantum attacks but become insecure when considering general 
so-called {\em non-signaling} provers, which are restricted {\em solely} 
by the requirement that no communication takes place.

This poses the natural question whether there exists a two-prover 
commitment scheme that is secure under the {\em sole} assumption that 
no communication takes place, and that does not rely on any further 
restriction of the information processing capabilities of the dishonest 
provers; no such scheme is known. 

In this work, we give strong evidence for a negative answer: 
%we answer this question in the negative: 
we show that any single-round 
two-prover commitment scheme can be broken by a non-signaling attack. 
Our negative result is as bad as it can get: for any candidate scheme 
that is (almost) perfectly hiding, there exists a strategy that allows 
the dishonest provers to open a commitment to an arbitrary bit (almost) 
as successfully as the honest provers can open an honestly prepared 
commitment, i.e., with probability (almost) $1$ in case of a perfectly 
sound scheme. 
In the case of multi-round schemes, our impossibility result is restricted to perfectly hiding schemes. 

On the positive side, we show that the impossibility result can be 
circumvented by considering {\em three} provers instead: there exists a 
three-prover commitment scheme that is secure against arbitrary 
non-signaling attacks.
%\keywords{non-signaling, bit-commitment, multi-prover}
\end{abstract}

%=================================================================
\section{Introduction}
%=================================================================

\paragraph{\bf Background. } A commitment scheme is an important primitive in
theoretical cryptography with various applications, for instance to
zero-knowledge proofs and multiparty computation, which themselves are
fundamentally important concepts in modern cryptography.  For a commitment
scheme to be secure, it must be {\em hiding} and {\em binding}. The former
means that after the commit phase, the committed value is still hidden from the
verifier, and the latter means that the prover (also referred to as committer)
can open a commitment only to one value.  Unfortunately, a commitment scheme
cannot be unconditionally hiding {\em and} unconditionally binding at the same
time. This is easy to see in the classical setting, and holds as well when using quantum communication~\cite{Mayers97,LC97}. Thus, we have to put some
limitation on
the capabilities of the dishonest party.  One common approach is to assume that
the dishonest prover (or, alternatively, the dishonest verifier) has limited
computing resources, so that he cannot solve certain computational problems
(like factoring large integers).  Another approach was suggested by Ben-Or,
Goldwasser, Kilian and Wigderson in their seminal paper~\cite{BGKW88} in the
late eighties. They assume that the prover consists of two (or more) agents
that cannot communicate with each other, and they show the existence of a
secure commitment scheme in this two-prover setting. Based on this two-prover
commitment scheme, they then show that every language in NP has a two-prover
perfect zero-knowledge interactive proof system (though there are some subtle
issues in this latter result, as discussed in~\cite{Yang13}). 

A simple example of a two-prover commitment scheme, due to~\cite{CSST11}, is the following. The verifier chooses a uniformly random string $a \in \set{0,1}^n$ and sends it to the first prover, who sends back $x := r \oplus a \cdot b$ as the commitment for bit $b \in \set{0,1}$, where $r \in \set{0,1}^n$ is a uniformly random string known (only) to the two provers, and where ``$\oplus$'' is bit-wise XOR and ``$\cdot$'' scalar multiplication (of the scalar $b$ with the vector~$a$). In order to open the commitment (to $b$), the second prover sends back $y := r$, and the verifier checks the obvious: whether $y = x \oplus a \cdot b$. It is clear that this scheme is hiding: $x := r \oplus a \cdot b$ is uniformly random and independent of $a$ no matter what $b$ is, and the intuition behind the binding property is the following. In order to open the commitment to $b = 0$, the second prover needs to announce $y = x$; in order to open to $b = 1$, he needs to announce $y = x \oplus a$. Therefore, in order to open to {\em both}, he must know $x$ {\em and} $x \oplus a$, which means he knows $a$, but this is a contradiction to the no-communication assumption, because $a$ was sent only to the first prover. 

In~\cite{CSST11}, Cr{\'e}peau, Salvail, Simard and Tapp show that, as a 
matter of fact, the security of such two-prover commitment schemes not 
only relies on the explicit assumption that the two provers cannot 
communicate, but the security also crucially depends on the information 
processing capabilities of the dishonest provers. Indeed, they show that a slight variation of the above two-prover commitment scheme (where some slack is given to the verification $y = x \oplus a \cdot b$) is secure against 
classical provers, but is completely insecure if the provers have 
{\em quantum} information processing capabilities and can obtain $x$ and $y$ by means of doing local measurements on an entangled quantum state.%
\footnote{The above intuition for the binding property of the scheme (which also applies to the variation considered in~\cite{CSST11}) fails in the quantum setting where $x$ and $y$ are obtained by means of {\em destructive} measurements.  }
Furthermore, they 
show that the above example two-prover commitment scheme remains secure 
against such quantum attacks, but becomes insecure against 
so-called {\em non-signaling} provers. 
The notion of non-signaling was first introduced by Khalfin and Tsirelson \cite{KT85} and by
Rastall \cite{Ras85} in the context of Bell-inequalities, and later
reintroduced by Popescu and Rohrlich \cite{PR94}. 
Non-signaling provers are 
restricted {\em solely} by the requirement that no communication takes 
place\,---\,no additional restriction limits their information 
processing capabilities (not even the laws of quantum mechanics)\,---\,and
thus considering non-signaling provers is the {\em minimal} assumption 
for the two-prover setting to make sense. 

This gives rise to the following question. Does there exist a two-prover 
commitment scheme that is secure against arbitrary non-signaling 
provers? Such a scheme would {\em truly} be based on the sole 
assumption that the provers cannot communicate. No such scheme is known. 
Clearly, from a practical point of view, asking for such a scheme may be 
overkill; given our strong believe in quantum mechanics, relying on a 
scheme that resists quantum attacks seems to be a safe bet.
But from a theoretical perspective, this question is certainly in line 
with the general goal of theoretical cryptography: to find the strongest 
possible security based on the weakest possible assumption.

\paragraph{\bf Our Results. }
In this work, we give strong evidence for a negative answer: 
%we answer the above question in the negative: 
we show that there exists no single-round two-prover commitment scheme that is secure against 
general non-signaling attacks. Our impossibility result is as strong as 
it can get. We show that for any candidate single-round two-prover commitment scheme that is 
(almost) perfectly hiding, the binding property can be (almost) 
completely broken: there exists a non-signaling strategy that allows the 
dishonest provers to open a commitment to an arbitrary bit (almost) 
as successfully as the honest provers can open an honestly prepared 
commitment, i.e., with probability (almost) $1$ in case of a perfectly 
sound scheme. Furthermore, for a restricted but natural class of 
schemes, namely for schemes that have the same communication pattern as the above example scheme, our impossibility result is tight: for every (rational) 
parameter $0 < \eps \leq 1$ there exists a perfectly sound two-prover 
commitment scheme that is $\eps$-hiding and as binding as allowed by our 
negative result (which is almost not binding if $\eps$ is small). 

In the case of multi-round schemes, our impossibility result is limited and applies to perfectly hiding schemes only. Proving the impossibility of non-perfectly-hiding multi-round schemes remains open. 

On the positive side, we show the existence of a secure {\em 
three}-prover commitment scheme against non-signaling attacks. Thus, our  impossibility result can be circumvented by considering three 
instead of two provers.

\paragraph{\bf Related Work. }

Two-prover commitments are closely related to {\em relativistic commitments}, as introduced by Kent in~\cite{Kent99}. In a nutshell, a relativistic commitment scheme is a two-prover commitment scheme where the no-communication requirement is enforced by having the actions of the two provers separated by a space-like interval, i.e., the provers are placed far enough apart, and the scheme is executed quickly enough, so that no communication can take place by the laws of special relativity. As such, our impossibility result immediately implies impossibility of relativistic commitment schemes of the form we consider (e.g., we do not consider quantum schemes) against general non-signaling attacks. 

Very generally speaking, and somewhat surprisingly, the (in)security of cryptographic primitives against non-signaling attacks may have an impact on more standard cryptographic settings, as was recently demonstrated by Kalai, Raz and Rothblum~\cite{KRR14}, who showed the (computational) security of a {\em delegation scheme} based on the security of an underlying multi-party interactive proof system against non-signaling (or statistically-close-to-non-signaling) adversaries. 

% Moreover, there is a surprising connection between the non-signaling setting
% and the seemingly unrelated area of \em delegation schemes\em . In a delegation
% scheme, Alice delegates some computation to Bob who returns the result along
% with a \em proof \em that this is indeed the result of the computation he was
% supposed to carry out. Of course, this only makes sense for Alice if the cost
% of verifying the proof is significantly lower than the cost of computing the
% function.  Aiello, Bhatt, Ostrovsky and Rajagopalan suggested a construction of
% such schemes from \em multi-party interactive proof \em (MIP) systems (or
% probabilistically checkable proofs) and \em fully homomorphic encryption \em
% \cite{ABOR00}. However, Dwork, Langberg, Naour, Nissim and Reingold proved that
% their construction is in general not secure \cite{DLNNR04}.  Recently, Kalai,
% Raz and Rothblum found that if the underlying MIP system is secure against
% non-signaling (or statistically-close-to-non-signaling) adversaries, then the
% construction by Aiello \etal . is secure \cite{KRR14}.

%=================================================================
\section{Preliminaries}
%=================================================================
%=================================================================
\subsection{(Conditional) Distributions}
%=================================================================

\newcommand{\R}{\Bbb{R}}

For the purpose of this work, a {\em (probability) distribution} is a 
function $p: {\cal X} \rightarrow \R$, $x \mapsto p(x)$, where $\cal X$ 
is a finite non-empty set, with the properties that $p(x) \geq 0$ for 
every $x \in \cal X$ and $\sum_{x \in \cal X} p(x) = 1$.
For any subset $\Lambda \subset \cal X$, $p(\Lambda)$ is naturally 
defined as $p(\Lambda) = \sum_{x \in \Lambda} p(x)$, and it holds that
\begin{equation}
\label{eq:pr_sum}
p(\Lambda) + p(\Gamma) = p(\Lambda \cup \Gamma) - p(\Lambda \cap \Gamma) 
\leq 1 + p(\Lambda \cap \Gamma)
\end{equation}
for all $\Lambda,\Gamma \subset \cal X$.
A probability distribution is {\em bipartite} if it is of the form $p: 
{\cal X} \times {\cal Y} \rightarrow \R$.
In case of such a bipartite distribution $p(x,y)$, probabilities like 
$p(x\!=\!y)$, $p(x\!=\!f(y))$, $p(x\!\neq\!y)$ etc. are naturally 
understood as
$$
p(x\!=\!y) = p(\Set{(x,y) \in {\cal X} \times {\cal Y}}{x = y}) = 
\sum_{x \in {\cal X}, y \in {\cal Y} \atop \text{s.t. } x = y} p(x,y)
$$
etc. Also, for a bipartite distribution  $p: {\cal X} \times {\cal Y} 
\rightarrow \R$, the {\em marginals} $p(x)$ and $p(y)$ are given by 
$p(x) = \sum_y p(x,y)$ and $p(y) = \sum_x p(x,y)$, respectively. We note 
that this notation may lead to an ambiguity when writing $p(w)$ for some 
$w \in {\cal X} \cap {\cal Y}$; we avoid this by writing $p(x\!=\!w)$ or 
$p(y\!=\!w)$ instead, which are naturally understood.
The above obviously extends to arbitrary {\em multipartite} 
distributions $p(x,y,z)$ etc.

A {\em conditional (probability) distribution} is a function $p: {\cal 
X} \times {\cal A} \rightarrow \R$, $(x,a) \mapsto p(x|a)$, for finite 
non-empty sets $\cal X$ and $\cal A$, such that for every fixed $a^* \in 
\cal A$, the function $p(x|a^*)$ is a probability distribution in the 
above sense, which we also write as $p(x|a\!=\!a^*)$. As such, the above 
naturally extends to bi- and multipartite conditional probability 
distributions; e.g., if $p(x,y|a,b)$ is a conditional distribution then 
$p(x|a,b)$, $p(y|a,b)$, $p(x\!=\!y|a,b)$ etc. are all naturally defined. 
However, we emphasize that for instance $p(x|a)$ is in general {\em 
not} well defined\,---\,unless the corresponding conditional distribution 
$p(b|a)$ is given, or unless $p(x|a,b)$ does not depend on $b$. 

\begin{remark}\label{rem:convention}
By convention, we write $p(x|a,b) = p(x|a)$ to express that $p(x|a,b)$ does not depend on $b$, i.e., that $p(x|a,b_1) = p(x|a,b_2)$ for all $b_1$ and $b_2$, and as such $p(x|a)$ {\em is} well defined and equals $p(x|a,b)$. 
\end{remark}

A distribution $\delta(x)$ over $\cal X$ is called a {\em Dirac} 
distribution if there exists $x^* \in \cal X$ so that $\delta(x\!=\!x^*) 
= 1$, and a conditional distribution $\delta(x|a)$ over $\cal X$ is 
called a conditional {\em Dirac} distribution if $\delta(x|a\!=\!a^*)$ 
is a {\em Dirac} distribution for every $a^* \in \cal A$, i.e., for 
every $a^* \in \cal A$ there exists $x^* \in \cal X$ so that 
$\delta(x\!=\!x^*|a\!=\!a^*) = 1$.

Note that we often abuse notation slightly and simply write $p(x)$ 
instead of $p: {\cal X} \rightarrow \R$, $x \mapsto p(x)$; furthermore, 
we may use $p$ for different distributions and distinguish between them 
by using different names for the variable, like when we consider the two 
marginals $p(x)$ and $p(y)$ of a bipartite distribution $p(x,y)$.
Finally, given two distributions $p(x_0)$ and $q(x_1)$ over the same set 
${\cal X}$ (and similarly if we use the above convention and denote them by 
$p(x_0)$ and $p(x_1)$ instead), we write $p(x_0) = q(x_1)$ to denote 
that $p(x_0\!=\!w) = q(x_1\!=\!w)$ for all $w \in {\cal X}$. In a 
corresponding way, equalities like $p(x_0,x'_0,y) = q(x_1,x'_1,y)$ 
should be understood; in situations where we feel it is helpful, we may 
clarify that ``$x_0$ is associated with $x_1$, and $x'_0$ with $x'_1$'';
similarly for conditional distributions.

%=================================================================
\subsection{Gluing Together Distributions}
%=================================================================

We recall the definition of the statistical distance.

\begin{definition}
Let $p(x_0)$ and $p(x_1)$ be two distributions over the same set ${\cal 
X}$.%
\footnote{This is without loss of generality: the domain can always be 
extended by including zero-probability elements. } Then, their 
statistical distance is defined as
$$
d\bigl(p(x_0),p(x_1)\bigr) = \frac{1}{2}\cdot \sum_{x\in\mathcal 
X}\bigl|p(x_0\!=\!x) - p(x_1\!=\!x)\bigr| \, .
$$
\end{definition}
The following property of the statistical distance is well known
(see~e.g. \cite{RK05}).

\begin{proposition}\label{prop:gluing}
Let $p(x_0)$ and $p(x_1)$ be two distributions over the same set ${\cal 
X}$ with $d\bigl(p(x_0),p(x_1)\bigr) = \varepsilon$. Then, there exists 
a distribution $p'(x_0,x_1)$ over ${\cal X} \times {\cal X}$ with 
marginals $p'(x_0) = p(x_0)$ and $p'(x_1) = p(x_1)$, and such that 
$p'(x_0 \!\neq\!x_1) = \varepsilon$.
\end{proposition}
The following is an immediate consequence.

\begin{lemma}\label{lemma:gluing}
Let $p(x_0,y_0)$ and $p(x_1,y_1)$ be distributions with 
$d\bigl(p(x_0),p(x_1)\bigr) = \eps$. Then, there exists a distribution 
$p'(x_0,x_1,y_0,y_1)$ with marginals $p'(x_0,y_0) = p(x_0,y_0)$ and 
$p'(x_1,y_1) = p(x_1,y_1)$, and such that $p'(x_0 \!\neq\!x_1) = 
\varepsilon$ and, as a consequence, 
$d\bigl(p'(x_0,y_1),p'(x_1,y_1)\bigr) \leq \eps$.
\end{lemma}

\begin{proof}
We first apply Proposition~\ref{prop:gluing} to $p(x_0)$ and $p(x_1)$ to 
obtain $p'(x_0,x_1)$, and then we set $$p'(x_0,x_1,y_0,y_1) = 
p'(x_0,x_1) \cdot p(y_0|x_0) \cdot p(y_1|x_1) \, .$$
The claims on the marginals and on $p'(x_0 \!\neq\!x_1)$ follow 
immediately, and for the last claim we note that
\begin{align*}
p'(x_0,y_1) &= p'(x_0 \!=\!x_1) \cdot p'(x_0,y_1|x_0 \!=\!x_1) + p'(x_0 
\!\neq\!x_1) \cdot p'(x_0,y_1|x_0 \!\neq\!x_1) \\
&= p'(x_0 \!=\!x_1) \cdot p'(x_1,y_1|x_0 \!=\!x_1) + p'(x_0 \!\neq\!x_1) 
\cdot p'(x_0,y_1|x_0 \!\neq\!x_1)
\end{align*}
and
\begin{align*}
p'(x_1,y_1) &= p'(x_0 \!=\!x_1) \cdot p'(x_1,y_1|x_0 \!=\!x_1) + p'(x_0 
\!\neq\!x_1) \cdot p'(x_1,y_1|x_0 \!\neq\!x_1)
\end{align*}
and the claim follows because $p'(x_0 \!\neq\!x_1) = \eps$.
\end{proof}

\begin{remark}
Note that due to the consistency of the marginals, it makes sense 
to write $p(x_0,x_1,y_0,y_1)$ instead of $p'(x_0,x_1,y_0,y_1)$.
We say that we ``glue together'' $p(x_0,y_0)$ and $p(x_1,y_1)$ along 
$x_0$ and $x_1$. 
\end{remark}

\begin{remark}
In the special case where $p(x_0)$ and $p(x_1)$ are identically distributed,
i.e., $d\bigl(p(x_0),p(x_1)\bigr) = 0$, we obviously have $p(x_0,y_1) = p(x_1,y_1)$.
\end{remark}

\begin{remark}\label{rem:commute}
It is easy to see from the proof of Lemma~\ref{lemma:gluing} that the following natural property holds. If $p(x_0,x_1,y_0,y_1,y'_0,y'_1)$ is obtained by gluing together $p(x_0,y_0,y'_0)$ and $p(x_1,y_1,y'_1)$ along $x_0$ and $x_1$, then the marginal $p(x_0,x_1,y_0,y_1)$ coincides with the distribution obtained by gluing together the marginals $p(x_0,y_0)$ and $p(x_1,y_1)$ along $x_0$ and $x_1$. 
% In other words, the actions of marginalizing and gluing together commute. 
\end{remark}

%=================================================================
\section{Bipartite Systems and Two-Prover Commitments}
%=================================================================

%=================================================================
\subsection{One-Round Bipartite Systems}
%=================================================================

Informally, a {\em bipartite system} consists of two subsystem, which we refer to as the left and the right subsystem. Upon input $a$ to the left and input $a'$ to the right subsystem, the left subsystem outputs $x$ and the right subsystem outputs~$x'$ (see Figure~\ref{fig:systems}, left). Formally, the behavior of such a system is given by a conditional distribution $q(x,x'|a,a')$, with the interpretation that upon input pair $(a,a')$, the system outputs a specific pair $(x,x')$ with probability $q(x,x'|a,a')$. Note that we leave the sets ${\cal A}, {\cal A}',{\cal X}$ and ${\cal X}'$, from which $a,a',x$ and $x'$ are respectively sampled, implicit. 

If we do not put any restriction upon the system, then {\em any} conditional distribution $q(x,x'|a,a')$ is eligible, i.e., describes a bipartite system. However, we are interested in systems where the two subsystems cannot communicate with each other. How exactly this requirement restricts $q(x,x'|a,a')$ depends on the available ``resources''. For instance, if the two subsystems are deterministic, i.e., compute $x$ and $x'$ as {\em deterministic} functions of $a$ and $a'$ respectively, then this restricts $q(x,x'|a,a')$ to be of the form $q(x,x'|a,a') = \delta(x|a) \cdot \delta(x'|a')$ for conditional Dirac distributions $\delta(x|a)$ and $\delta(x'|a')$. If in addition to allowing them to compute deterministic functions, we give the two subsystem {\em shared randomness}, then $q(x,x'|a,a')$ may be of the form 
$$
q(x,x'|a,a')  = \sum_r p(r) \cdot \delta(x|a,r) \cdot \delta(x'|a',r)
$$
for a distribution $p(r)$ and conditional Dirac distributions $\delta(x|a,r)$ and $\delta(x'|a',r)$. Such a system is called {\em classical} or {\em local}. 
Interestingly, this is not the end of the story. By the laws of {\em quantum mechanics}, if the two subsystems share an entangled quantum state and obtain $x$ and $x'$ without communication as the result of local measurements that may depend on $a$ and $a'$, respectively, then this gives rise to conditional distributions $q(x,x'|a,a')$ of the form 
$$
q(x,x'|a,a') = \big\langle \psi\big| \big(E_x^a \otimes F_{x'}^{a'}\big) \big|\psi\big\rangle \, ,
$$
where $|\psi\rangle$ is a quantum state and $\set{E_x^a}_x$ and $\set{F_{x'}^{a'}}_{x'}$ are so-called POVMs. What this exactly means is not important for us; what {\em is} important is that this leads to a {\em strictly larger} class of bipartite systems. This is typically referred to as a {\em violation of Bell inequalities}~\cite{bell}, and is nicely captured by the notion of {\em nonlocal games}. A famous example is the so-called CHSH-game~\cite{chsh}, which is closely connected to the example two-prover commitment scheme from the introduction, and which shows that the variant considered in~\cite{CSST11} is insecure against quantum attacks.

The largest possible class of bipartite systems that is compatible with the requirement that the two subsystem do not communicate, but otherwise does not assume anything on the available resources and/or the underlying physical theory, are the so-called {\em non-signaling} systems, defined as follows. 

\begin{definition}\label{def:ns}
A conditional distribution $q(x,x'|a,a')$ is called a {\em non-signaling (one-round) bipartite system} if it satisfies
$$
q(x|a,a') = q(x|a) \qquad\text{\rm (NS)}
$$
as well as with the roles of the primed and unprimed variables exchanged, i.e., 
$$
q(x'|a,a') = q(x'|a') \qquad\text{\rm (NS$'$)}
$$
\end{definition}
Recall that, by the convention in Remark~\ref{rem:convention}, the equality (NS) is to be understood in the sense that $q(x|a,a')$ does not depend on $a'$, i.e., that $q(x|a,a'_1) = q(x|a,a'_2)$ for all $a'_1,a'_2$, and correspondingly for (NS$'$). 

We emphasize that this is the {\em minimal} necessary condition for the requirement that the two subsystems do not communicate. Indeed, if e.g.~$q(x|a,a'_1) \neq q(x|a,a'_2)$, i.e., if the input-output behavior of the left subsystem depends on the input to the right subsystem, then the system can be used to communicate by giving input $a'_1$ or $a'_2$ to the right subsystem, and observing the input-output behavior of the left subsystem. Thus, in such a system, communication does take place. 

% If a conditional distribution $q(x,x'|a,a')$ satisfies the above non-signaling condition (NS), then the conditional distribution $q(x|a)$ is well defined as $q(x|a) = q(x|a,a')$ for an arbitrary $a'$, and it holds that $q(x|a,a') = q(x|a)$ for {\em every} $a'$. On the other hand, if $q(x,x'|a,a')$ does not satisfy (NS), then $q(x|a)$ is {\em not} well defined but depends on the (conditional) distribution of $a'$: $q(x|a) = \sum_{a'} p(a'|a) \, q(x|a,a')$. Therefore, we slightly abuse notation and re-write the non-signaling condition (NS) as 
% $$
% q(x|a,a') = q(x|a) \qquad\text{\rm (NS)}
% $$ 
% and take it as understood that this expression requires that $q(x|a)$ is well defined. In other words, we use $q(x|a,a') = q(x|a)$ as a synonym for saying that ``$q(x|a,a')$ does not depend on $a'$''. Similarly for (NS$'$). 

The non-signaling requirement for a bipartite system is\,---\,conceptually and formally\,---\,equivalent to requiring that the two subsystems can (in principle) be queried {\em in any order}. 
Conceptually, it holds because the left subsystem should be able to deliver its outputs {\em before} the right subsystem has received any input if and only if the output does not depend on the right subsystem's input (which means that no information is communicated from right to left), and similarly the other way round. 
% Conceptually, it holds because if, say, the system communicates from the right to the left, then, due to causality, it is not possible that the left subsystem delivers its output when the right system has not been given its input yet. 
And, formally, we see that
% The non-signaling condition is equivalent to requiring that the two subsystems can\,---\,in principle\,---\,be queried in any order, and that the subsystem that is queried first can produce its (correctly distributed) output before the input to the other subsystem is chosen. Indeed, 
the non-signaling requirement from Definition~\ref{def:ns} is equivalent to asking that $q(x,x'|a,a')$ can be written as 
$$
q(x,x'|a,a') = q(x|a) \cdot q(x'|x,a,a')
\quad\text{and}\quad
q(x,x'|a,a') = q(x'|a') \cdot q(x|x',a,a')
$$
for some respective conditional distributions $q(x|a)$ and $q(x'|a')$.
This characterization is a convenient way to ``test'' whether a given bipartite system is non-signaling without doing the maths. 

Clearly, all classical systems are non-signaling. Also, any quantum system is non-signaling.%
\footnote{Indeed, the two parts of an entangled quantum state can be measured in any order, and the outcome of the first measurement does not depend on how the other part is going to be measured.}
But there are non-signaling systems that are not quantum (and thus in particular not classical). The typical example is the {\em NL-box} (non-local box; also known as {\em PR-box})~\cite{PR94}, which, upon input bits $a$ and $a'$ outputs {\em random} output bits $x$ and $x'$ subject to
$$
x \oplus x' = a \cdot a' \, .
$$
This system is indeed non-signaling, as it can be queried in any order: submit $a$ to the left subsystem to obtain a uniformly random $x$, and then submit $a'$ to the right subsystem to obtain $x' := x \oplus a \cdot b$, and correspondingly the other way round. 

%=================================================================
\subsection{Two-Round Systems}\label{sec:2round}
%=================================================================

We now consider bipartite systems as discussed above, but where one can interact with the two subsystems multiple times. We restrict to two rounds: after having input $a$ to the left subsystem and obtained $x$ as output, one can now input $b$ into the left subsystem and obtain output $y$, and similarly with the right subsystem (see Figure~\ref{fig:systems}, right). 
In such a two-round setting, the non-signaling condition needs to be paired with {\em causality}, which captures that the output of the first round does not depend on the input that will be given in the second round.  

\begin{figure}[h]
  \centering
  \scalebox{0.8}{ \input{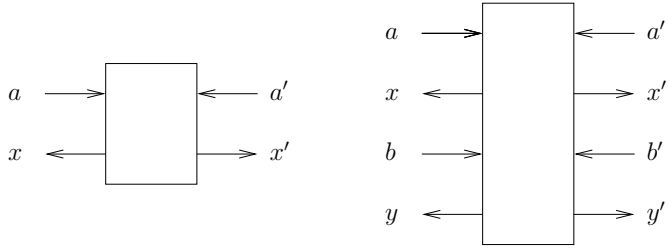} }
  \caption{A one-round (left) and two-round (right) bipartite system.}\label{fig:systems}
  \label{fig:OneTwoRound}
\end{figure}

\begin{definition}\label{def:ns2}
A conditional distribution $q(x,x',y,y'|a,a',b,b')$ is called a {\em non-signaling two-round bipartite system} if it satisfies the following two causality constraints
\begin{align*}
&q(x,x'|a,a',b,b') = q(x,x'|a,a') &\text{\rm (C1)} \\
\text{and}\quad  &q(x'|x,y,a,a',b,b') = q(x'|x,y,a,a',b) &\text{\rm (C2)}
\end{align*}
and the following two non-signaling constraints
\begin{align*}
&q(x,y|a,a',b,b') = q(x,y|a,b) &\text{\rm (NS1)}\\
\text{and}\quad &q(y|x,x',a,a',b,b') = q(y|x,x',a,a',b) &\quad\text{\rm (NS2)}
\end{align*}
as well as with the roles of the primed and unprimed variables exchanged. 
\end{definition}
(C1) captures causality of the overall system, i.e., when considering the left and the right system as one ``big'' multi-round system. (C2)~captures that no matter what interaction there is with the left system, the right system still satisfies causality. Similarly, (NS1) captures that the left and the right system are non-signaling over both rounds, and (NS2) captures that no matter what interaction there was in the first round, the left and the right system remain non-signaling in the second round. 

It is rather clear that these are {\em necessary} conditions; we argue that they are {\em sufficient} to capture a non-signaling two-round system in Appendix~\ref{app:sufficient}.

%=================================================================
\subsection{Two-Prover Commitments}
%=================================================================

We consider two-prover commitments of the following form. To commit to bit $b$, the two provers $\hP$ and $\hQ$ receive respective ``questions" $a$ and $a'$ from the verifier $V$, and they compute, without communicating with each other, respective replies $x$ and $x'$ and send them to $V$. 
To open the commitment, $\hP$ and $\hQ$ send respectively $y$ and $y'$. Finally, $V$ performs some check to decide whether to accept or not.  

In case of {\em classical} provers $\hP$ and $\hQ$, restricting the opening phase to one round with one-way communication is without loss of generality: one may always assume that in the opening phase $\hP$ and $\hQ$ simply reveal the shared randomness, and $V$ checks whether $x$ and $x'$ had been correctly computed, consistent with the claimed bit $b$. 
Restricting the commit phase to one round is, as far as we can see, {\em not} without loss of generality; we discuss the multi-round case later. 

Formally, this can be captured as follows.   

\begin{definition}\label{def:Com}
A {\em (single-round) two-prover commitment scheme} $\Com$ consists of a probability distribution $p(a,a')$, two conditional distributions $p_0(x,x',y,y'|a,a')$ and $p_1(x,x',y,y'|a,a')$, and an acceptance predicate $\Acc(x,x',y,y'|a,a',b)$. \\
We say that $\Com$ is classical/quantum/non-signaling if 
$p_0(x,x',y,y'|a,a')$ and $p_1(x,x',y,y'|a,a')$ are both 
classical/quantum/non-signaling when parsed as bipartite one-round systems $p_b((x,y),(x',y')|a,a')$. By default, any two-prover commitment scheme $\Com$ is assumed to be non-signaling. 
\end{definition}
The distribution $p(a,a')$ captures how $V$ samples the ``questions'' $a$ and $a'$, $p_b(x,x',y,y'|a,a')$ describes the choices of $x$ and $x'$ and of $y$ and $y'$, given that the bit to commit to is $b$, and $\Acc(x,x',y,y'|a,a',b)$ determines whether $V$ accepts the opening or not. 
% In this paper, we mostly consider non-signaling schemes and their security
% against non-signaling adversaries. 
Whether a scheme is classical, quantum or non-signaling captures the restrictions of the honest provers. 

Given a two-prover commitment scheme
$\Com$, we define
$$
\Prob[\Acc|b] := \sum_{a,a',x,x',y,y'}  p(a,a') \cdot p_b(x,x',y,y'|a,a') \cdot \Acc(x,x',y,y'|a,a',b) \, ,
$$
which is the probability that a correctly formed commitment to bit $b$ is successfully opened. 
%Informally, we say that such a commitment scheme is {\em sound} if $\Prob[\Acc|b]$ is close to $1$ for both $b = 0$ and $b = 1$. 

\begin{definition} A commitment scheme $\Com$ is $\theta$-\em sound
\em if $\Prob _p[\Acc |b] \geq \theta$ for $b \in \set{0,1}$. We say that it is
\em perfectly sound \em if it is $1$-sound.  \end{definition}
It will be convenient to write $p(x_0,x_0',y_0,y_0'|a,a')$ instead of
$p_0(x,x',y,y'|a,a')$ and $p(x_1,x_1',y_1,y_1'|a,a')$ instead of
$p_1(x,x',y,y'|a,a')$. Switching to this notation, the hiding property is
expressed as follows.  \begin{definition} $\Com$ is called $\eps$-{\em hiding}
if $d\bigr(p(x_0,x_0'|a,a'),p(x_1,x_1'|a,a')\bigl) \leq \eps$ for
all $a,a'$. If $\Com$ is $0$-hiding, we also say it is \em perfectly hiding\em
.  \end{definition}
Capturing the binding property is more subtle. From the classical approach of
defining the binding property for a commitment scheme, one is tempted to
require that once the commit phase is over and $a,a',x$ and $x'$ are fixed,
adversarial provers $\dP$ and $\dQ$ cannot come up with an opening to $b = 0$
and {\em simultaneously} with an opening to $b = 1$, i.e., with $y_0,y'_0$  and
$y_1,y'_1$ such that $\Acc(x,x',y_0,y'_0|a,a',b\=0)$ and
$\Acc(x,x',y_1,y'_1|a,a',b\=1)$ are both satisfied (except with small
probability).  However, as pointed out by Dumais, Mayers and Salvail
\cite{DMS00}, in the context of a
general physical theory where $y$ and $y'$ may possibly be obtained as
respective outcomes of {\em destructive} measurements (as is the case in
quantum mechanics), such a definition is too weak. It does not exclude that
$\dP$ and $\dQ$ can freely choose to open the commitment to $b = 0$ or to $b =
1$, whatever they want, but they cannot do both {\em simultaneously}; once they
have produced one opening, their respective states got disturbed and the other
opening can then not be obtained anymore.  

Our definition for the binding property is based on the following game between the (honest) verifier $V$ and the adversarial provers $\dP$, $\dQ$. 

\begin{enumerate}\setlength{\parskip}{0.5ex}
\item The commit phase is executed: $V$ samples $a$ and $a'$ according to $p(a,a')$, and sends $a$ to $\dP$ and $a'$ to~$\dQ$, upon which $\dP$ and $\dQ$ send $x$ and $x'$ back to $V$, respectively.
\item $V$ sends a bit $b \in \set{0,1}$ to $\dP$ and $\dQ$.
\item $\dP$ and $\dQ$ try to open the commitment to $b$: they prepare $y$ and $y'$ and send them to $V$.
\item $V$ checks if the verification predicate $\Acc(x,x',y,y'|a,a',b)$ is satisfied. 
\end{enumerate}

We emphasize that even though in the actual binding game above, {\em the same}
bit $b$ is given to the two provers, we require that the response of the
provers is well determined by their strategy even in the case that
$b \neq b'$. Of course, if the provers are allowed to communicate, they are able
to detect when $b \neq b'$ and could reply with, e.g., $y = y' = \bot$ in that
case. However, if we restrict to non-signaling provers, we assume that it is
\em physically \em impossible for them to communicate with each other and
distinguish the case of $b = b'$ from $b\neq b'$.

As such, a non-signaling attack strategy against the binding property of a two-prover commitment scheme $\Com$ is given by a non-signaling two-round bipartite system $q(x,x',y,y'|a,a',b,b')$, as specified in Definition~\ref{def:ns2}. For any such bipartite system, representing a strategy for $\dP$ and $\dQ$ in the above game, the probability that $\dP$ and $\dQ$ win the game, in that $\Acc(x,x',y,y'|a,a',b)$ is satisfied when they have to open to the bit $b$, is given by
$$
\dProb_{q}[\Acc|b] := \sum_{a,a',x,x',y,y'}  p(a,a') \cdot q(x,x',y,y'|a,a',b,b) \cdot \Acc(x,x',y,y'|a,a',b) \, .
$$
We are now ready to define the binding property. 
\begin{definition}\label{def:binding}
A two-prover commitment scheme $\Com$ is $\delta$-{\em binding} (against non-signaling attacks) if it holds for any non-signaling two-round bipartite system $q(x,x',y,y'|a,a',b,b')$ that
$$
\dProb_{q}[\Acc|0] + \dProb_{q}[\Acc|1] \leq 1 + \delta \, .
$$
\end{definition}
% In this context, we also refer to a non-signaling two-round bipartite system $q(x,x',y,y'|a,a',b,b')$ as a (non-signaling two-prover) {\em strategy}. 
In other words, a scheme is $\delta$-binding if in the above game the dishonest provers win with probability at most $(1+\delta)/2$ when $b \in \set{0,1}$ is chosen uniformly at random. 
% We could equivalently define a scheme as $\delta$-binding if every
% non-signaling strategy $q(x,x',y,y'|a,a',b,b')$ satisfies the
% $\Acc$-predicate with probability at most $(1+\delta)/2$ when $a,a'$ is sampled
% according to $p(a,a')$ and $b = b'\in\set{0,1}$ is selected uniformly at random.
If a commitment scheme is binding (for a small $\delta$) in the sense of Definition~\ref{def:binding}, then for any strategy $q$ for $\dP$ and $\dQ$, they can just as well {\em honestly} commit to a bit $\hat{b}$, where $\hat{b}$ is set to $0$ with probability $p_0 = \dProb_{q}[\Acc|0]$ and to $1$ with probability $p_1 = 1 - p_0 \approx \dProb_{q}[\Acc|1]$, and they will have essentially the same respective success probabilities in opening the commitment to $b = 0$ and to $b = 1$. 

% We stress once more that even though in the actual binding game $\dP$ and $\dQ$ are given the same bit, it must be well defined how they behave when given different bits, and it is crucial to consider this case to be able to say whether they are non-signaling or not. 

%=================================================================
\section{Impossibility of Two-Prover Commitments}
%=================================================================

In this section, we show impossibility of secure single-round two-prover commitments against arbitrary non-signaling attacks. We start with the analysis of a restricted class of schemes which are easier to understand and for which we obtained stronger results.

%=================================================================
\subsection{Simple Schemes}
\label{sec:simple}
%=================================================================
We first consider a special, yet natural, class of schemes. 
We call a two-prover commitment scheme $\Com$ {\em simple} if it has the same communication pattern as the scheme described in the introduction. More formally, it is called simple if $a',x'$ 
and $y$ are ``empty'' (or fixed), i.e., if $\Com$ is given by $p(a)$, 
$p_0(x,y'|a)$, $p_1(x,y'|a)$ and $\Acc(x,y'|a,b)$; to simplify notation, 
we then write $y$ instead of~$y'$. In other words, $\hP$ is only 
involved in the commit phase, where, in order to commit to bit $b$, he 
outputs $x$ upon input $a$, and $\hQ$ is only involved in the opening 
phase, where he outputs $y$. The non-signaling requirement for $\Com$ then simplifies to $p_b(y|a) = p_b(y)$. 
Recall that by our convention, we may write $p(x_0,y_0|a)$ instead of $p_0(x,y|a)$ and $p(x_1,y_1|a)$ instead of $p_1(x,y|a)$.

In case of such a simple two-prover commitment scheme $\Com$, a 
non-signaling two-prover strategy reduces to a non-signaling {\em 
one-round} bipartite system as specified in 
Definition~\ref{def:ns} (see Figure~\ref{fig:Simple}).

\begin{figure}[h]
  \centering
  \scalebox{0.8}{ \input{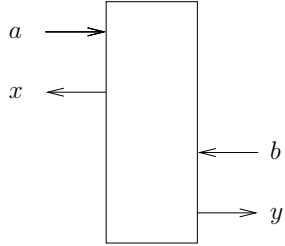} }
  \caption{The adversaries' strategy $q(x,y|a,b)$ in case of a {\em simple} commitment scheme. }
  \label{fig:Simple}
\end{figure}

% An example of a simple two-prover commitment scheme, which is due to 
% Cr{\'e}peau \etal.~\cite{CSST11}, is the following. $\hP$ and $\hQ$ hold 
% a random $n$-bit string $r \in \set{0,1}^n$ as shared randomness. In 
% order to commit to bit $b \in \set{0,1}$, $\hP$ sends $x = r \oplus a 
% \cdot b$ to $V$ upon receiving random $a \in \set{0,1}^n$ from $V$, 
% i.e., $x = r$ if $b = 0$ and $x = r \oplus a$ if $b = 1$. And in order 
% to open the commitment (to $b$), $\hQ$ sends $y = r$ to $V$, who 
% verifies that $x \oplus y = a \cdot b$.
% It was shown in~\cite{CSST11} that this particular scheme is secure 
% against classical and against quantum attacks, but not against arbitrary 
% non-signaling attacks.; indeed, it is easy to see that an NL-box breaks 
% the scheme. And it is a variant of this scheme, where the verifier 
% accepts the opening if $x \oplus y$ is ``close enough'' to $a \cdot b$, 
% that is secure against classical attacks but not against quantum attacks.

As a warm-up exercise, we first consider a simple two-prover 
commitment scheme that is {\em perfectly hiding} and {\em perfectly sound}. Recall that formally, a 
simple scheme is given by $p(a)$, $p_0(x,y|b)$, $p_1(x,y|a)$ and 
$\Acc(x,y|a,b)$, and the perfect hiding property means that $p_0(x|a) = 
p_1(x|a)$ for any $a$. To show that such a scheme cannot be binding, we 
have to show that there exists a non-signaling one-round bipartite 
system $q(x,y|a,b)$ such that $\dProb_{q}[\Acc|0] + 
\dProb_{q}[\Acc|1]$ is significantly larger than $1$. But this is 
actually trivial: we can simply set $q(x,y|a,b) := p_b(x,y|a)$. It then 
holds trivially that 
\begin{align*}
\dProb_{q}[\Acc|b] &= \sum_{a,x,y} p(a) \, q(x,y|a,b) \,\Acc(x,y|a,b)\\
  &= \sum_{a,x,y} p(a) \,p_b(x,y|a)\, \Acc(x,y|a,b)\\
  &= \Prob_{p}[\Acc|b]
\end{align*}
and thus 
that the dishonest provers are as successful in opening the commitment 
as are the honest provers in opening an honestly prepared commitment. Thus, 
the binding property is broken as badly as it can get. The only thing 
that needs to be verified is that $q(x,y|a,b)$ is non-signaling, i.e., 
that $q(x|a,b) = q(x|a)$ and $q(y|a,b) = q(y|b)$. To see that the latter holds,
note that $q(y|a,b) = p_b(y|a)$, and because $\Com$ is non-signaling we have that $p_b(y|a) = p_b(y)$, i.e., does not depend on $a$. Thus, the same holds for $q(y|a,b)$ and we have $q(y|a,b) = q(y|b)$.
The former condition follows
from the (perfect) hiding property: $q(x|a,b) = p_b(x|a) = p_{b'}(x|a) = 
q(x|a,b')$ for arbitrary $b,b' \in \set{0,1}$, and thus $q(x|a,b) = 
q(x|a)$. 

Below, we show how to extend this result to non-perfectly-binding simple 
schemes. In this case, we cannot simply set $q(x,y|a,b) := p_b(x,y|a)$, 
because such a $q$ would not be non-signaling anymore\,---\,it would 
merely be ``almost non-signaling''. Instead, we have to find a strategy 
$q(x,y|a,b)$ that is (perfectly) non-signaling and close to 
$p_b(x,y|a)$; we will find such a strategy with the help of
Lemma~\ref{lemma:gluing}. 
In Section~\ref{sec:arbitrary}, we will then consider general schemes
where {\em both} provers interact with the verifier in {\em both} phases. In
this general case, further complications arise.

% We now show that a simple scheme that is perfectly sound and $\eps$-hiding
% can at best be $(1-\eps )$-binding. We further show that this bound is tight
% for all rational $\eps$ with $0 < \eps \leq 1$. In fact, this can be achieved
% with a \em classical \em scheme.

\begin{theorem}
\label{thm:tight}
Consider a simple two-prover commitment scheme $\Com$ 
% given by probability distributions
% $p(a)$, $p(x_0,y_0|a)$ and $p(x_1,y_1|a)$ and acceptance predicate
% $\Acc (x,y|a,b)$.  Suppose 
that is $\eps$-hiding. Then, there exists
a non-signaling strategy $q(x,y|a,b)$ such that
$$
  \Prob _q^* [\Acc|0] = \Prob _p [\Acc |0] \quad\text{and}\quad
  \Prob _q^* [\Acc|1] \geq \Prob _p [\Acc |1] - \eps \, .
$$
If $\Com$ is perfectly sound, it follows that
$$\Prob _q^*[\Acc|0]+\Prob _q^*[\Acc|1] \geq 1 +(1-\eps)$$
and thus it cannot be $\delta$-binding for $\delta < 1-\eps$.
\end{theorem}

\begin{proof}
Recall that $\Com$ is given by $p(a)$, $p_b(x,y|a)$ and $\Acc (x,y|a,b)$, and we write $p(x_b,y_b|a)$ instead of $p_b(x,y|a)$. 
Because $\Com$ is $\eps$-hiding, it holds that $d\bigl(p(x_0|a),p(x_1|a)\bigr) \leq \eps$ for
any fixed $a$. Thus, using Lemma~\ref{lemma:gluing} for every $a$, we can glue
together $p(x_0,y_0|a)$ and $p(x_1,y_1|a)$ along $x_0$ and $x_1$ to obtain a
distribution $p(x_0,x_1,y_0,y_1|a)$ such that $p(x_0\neq x_1|a)\leq\eps$, and
in particular $d\bigl(p(x_0,y_1|a), p(x_1,y_1|a)\bigr)\leq \eps$.

We define a strategy $q$ for the dishonest provers by setting
$q(x,y|a,b) := p(x_0,y_b|a)$ (see Figure \ref{fig:0and1}).
First, we show that $q$ is non-signaling. Indeed, we have
$q(x|a,b) = p(x_0|a)$ for any $b$, so $q(x|a,b) = q(x|a)$, and we have $q(y|a,b) = p(y_b|a) = p(y_b)$ for any $a$, and thus $q(y|a,b) = q(y|b)$. 

As for the acceptance probability, for $b = 0$ we have $q(x,y|a,0) = p(x_0, y_0|a)$ and as such \smash{$\dProb_q [\Acc |0]$} equals $\Prob _p [\Acc | 0]$. 
For $b=1$, we have
$$d\bigl(q(x,y|a,1),p(x_1,y_1|a)\bigr) = d\bigl(p(x_0,y_1|a),p(x_1,y_1|a)\bigr)\leq \eps$$
and since the statistical distance does not increase under data processing,
it follows that
$\Prob _p[\Acc |1]$ and $\Prob _q^*[\Acc |1]$ are $\eps$-close; this proves the claim.
\end{proof}

\begin{figure}[h]
  \centering
  \scalebox{0.8}{ \input{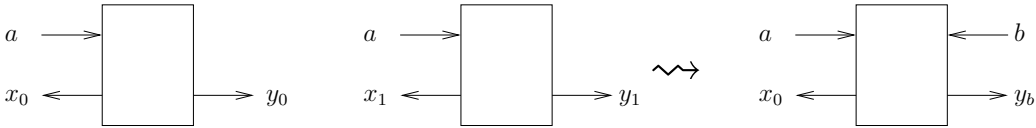} }
  \caption{Defining the strategy $q$ by gluing together $p(x_0,y_0|a)$ and
  $p(x_1,y_1|a)$.}
  \label{fig:0and1}
\end{figure}

The bound on the binding property in Theorem~\ref{thm:tight} is tight, as the following theorem shows.
\begin{theorem}
For all $\eps \in \Q$ such that $0 < \eps\leq 1$ there exists a classical simple two-prover commitment scheme that is perfectly sound, $\eps$-hiding and $(1-\eps)$-binding against non-signaling adversaries.
\end{theorem}

\begin{proof}
We construct a scheme where the first prover reveals the bit $b$
right at the beginning with probability $\eps$. 
For simplicity, we first assume that
$\eps = 1/n$ for some integer $n\geq 1$ and then indicate how to
extend the proof to arbitrary rational numbers.

The scheme works as follows. Let $[n] = \{ 0,\dots , n-1\}$. The shared
randomness of the provers is $r\in [n]$ selected uniformly at random.
The verifier selects $a\in [n]$ uniformly at random and sends it to prover $P$.
If $a = r$ then $P$ reveals $x := b$ to the verifier. Otherwise,
he sends back $x:= \bot$. In the opening phase, $Q$ sends $r$ to the
verifier. The verifier accepts if and only if $P$ revealed $b$ or the
output $y$ of $Q$ satisfies $y\in [n]$ and $y\neq a$.

It is clear that this scheme is sound and $\eps$-hiding. Now consider
dishonest provers that follow some non-signaling strategy $q(x,y|a,b)$. This then defines $q(a,x,y|b) = p(a) \, q(x,y|a,b)$ with $p(a) = 1/n$, and it holds that $\Prob _q^*[\Acc | b] = q(x\!=\!b|b) + q(x\!=\!\bot , y\neq a|b)$.
% Without loss of generality, we can restrict $x$ to
% be in $\{ 0, 1, \bot \}$ and $y$ to be in $[n]$.
Since $q(y|a,b) = q(y|b)$, %and $a\in [n]$ is selected uniformly at random and independent of $b$, 
we have 
$$
q(y\!\neq\! a|b) = \sum_{a,y \atop a \neq y} q(a,y|b) = \sum_{a,y \atop a \neq y} p(a) \, q(y|b) = \sum_y \frac{n-1}{n} q(y|b) = 1-\eps \, .
$$
% The provers win if and only if either $P$ sends the correct bit in the
% commit phase or $P$ sends $\bot$ and $Q$ sends some $y\neq a$, and thus 
% $\Prob _q^*[\Acc | b] = q(x\!=\!b|b) + q(x\!=\!\bot , y\neq a|b)$.
% We can upper-bound $q(x\!=\!\bot,y\!\neq\! a|b)$ in two different ways:
% \begin{samepage}
% \begin{align*}
%   q(x\!=\!\bot,y\!\neq\! a|a,b) &\leq q(x\!=\!\bot|a,b)\\
%   q(x\!=\!\bot,y\!\neq\! a|a,b) &\leq q(y\!\neq\! a|a,b) = 1-\eps
% \end{align*}
% \end{samepage}
% Using these upper bounds and the fact that $q(x|a,b) = q(x|a)$, we calculate
Therefore, using that $q(x|a,b) = q(x|a)$ and hence $q(x|b\!=\!0) = q(x|b\!=\!1)$, we calculate
\begin{align*}
  \Prob_q^*&[\Acc|0] + \Prob_q^*[\Acc|1]\\
  =\ & q(x\!=\!0|b\!=\!0)+q(x\!=\!\bot,y\neq a|b\!=\!0)+q(x\!=\!1|b\!=\!1)
    +q(x\!=\!\bot,y\neq a|b\!=\!1)\\
  \leq\ & q(x\!=\!0|b\!=\!0)+q(x\!=\!1|b\!=\!0)+q(x\!=\!\bot|b\!=\!0)+q(y\!\neq\! a|b\!=\!1) \\
  =\ & 1 + (1-\eps) \, .
\end{align*}
We now adapt this argument to $\eps = m/n$, where $m$ and $n$ are
integers such that $0 < m \leq n$. For every $a\in [n]$, we define
a subset $S_a$ of $[n]$ as
$$S_a = \{ a+i\bmod n \mid i\in \{ 0,\dots ,m-1\}\}\text .$$
We adapt our scheme by replacing the condition $r=a$ with $r\in S_a$.
Clearly, the scheme is still sound. Since every $S_a$ has exactly $m$ elements,
the scheme is $\eps$-hiding: the probability that the first prover
reveals $b$ is $m/n = \eps$; otherwise, he does not give any information
about $b$. The proof that the scheme is $(1-\eps)$-binding goes
through as before if we can show that $q(y\not\in S_a | a,b) = 1-\eps$
for any non-signaling strategy $q$. Indeed, for every $y\in [n]$, there
are exactly $m$ values for $a$ such that $y\in S_a$. Since $a\in [n]$ is
selected randomly and $q(y|a,b)$ is independent of $a$, we have
$q(y\not\in S_a|a,b) = 1 - m/n = 1-\eps$.
\end{proof}

%=================================================================
\subsection{Arbitrary Schemes}
\label{sec:arbitrary}
%=================================================================

We now remove the restriction on the scheme to be simple. As before, we first consider the case of a perfectly hiding scheme. 

\begin{theorem}\label{thm:impossibility}
Let $\Com$ be a single-round two-prover commitment scheme. 
If $\Com$ is perfectly hiding, then there exists a non-signaling two-prover strategy  $q(x,x',y,y'|a,a',b,b')$ such that 
$$
\dProb_{q}[\Acc|b] = \Prob_p[\Acc|b]
$$
for $b \in \set{0,1}$. 
\end{theorem}

\begin{proof}
$\Com$ being perfectly hiding means that
$d(p(x_0,x_0'|a,a'),p(x_1,x_1'|a,a'))=0$ for all $a$ and $a'$.
Gluing together the distributions $p(x_0,x_0',y_0,y'_0|a,a')$ and $p(x_1,x_1',y_1,y'_1|a,a')$ along $(x_0,x_0')$ and $(x_1,x_1')$ for every $(a,a')$,
we obtain a distribution $p(x_0,x_0',x_1,x_1',y_0,y_0',y_1,y_1'|a,a')$ with the correct
marginals and $p((x_0,x_0')\neq (x_1,x_1')|a,a') = 0$. That is, we have
$x_0 = x_1$ and $x_0' = x_1'$ with certainty. We now define a strategy for dishonest provers
as (Figure~\ref{fig:0and1_gen})
$$q(x,x',y,y'|a,a',b,b') := p(x_0,x_0',y_b,y_{b'}'|a,a') \, .
$$
Since $p(x_0,x_0',y_b,y_b'|a,a') = p(x_b,x_b',y_b,y_b'|a,a')$,
it holds that $\Prob _q^* [ \Acc | b] = \Prob _p [ \Acc | b]$.
It remains to show that this distribution satisfies the non-signaling and
causality constraints (C1) up to (NS2) of Definition~\ref{def:ns2}.
This is done below. 

\begin{itemize}\setlength{\parskip}{1.5ex}
\item For (C1), note that summing up over $y$ and $y'$ yields
  $q(x,x'|a,a',b,b') = p(x_0,x_0'|a,a')$, which indeed does not depend on
  $b$ and $b'$. 
\item For (NS1), note that
  $q(x,y|a,a',b,b') = p(x_0,y_b|a,a') = p(x_b,y_b|a,a') = p(x_b,y_b|a)$, where the last equality holds by the non-signaling property of $p(x_b,y_b|a,a')$.
\item For (C2), first note that
  \begin{equation}\label{eq:proof}
    q(x,x',y|a,a',b,b') = p(x_0,x_0',y_b|a,a')
  \end{equation}
  which does not depend on $b'$. We then see that (C2) holds by dividing by
  $q(x,y|a,a',b,b') = p(x_0,y_b|a,a')$. 
\item For (NS2), divide Equation \eqref{eq:proof} by
  $q(x,x'|a,a',b,b') = p(x_0,x_0'|a,a')$
\end{itemize}
The properties (C1) to (NS2) with the roles of the primed and unprimed
variables exchanged follows from symmetry. This concludes the proof. 
\end{proof}

\begin{figure}[h]
  \centering
  \scalebox{0.8}{ \input{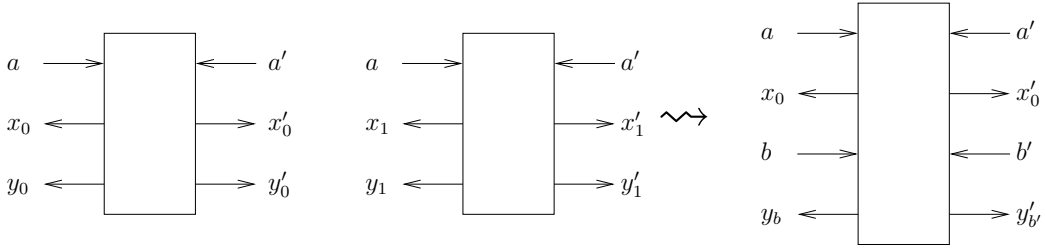} }
  \caption{Defining the strategy $q$ from $p(x_0,x_0',y_0,y_0'|a,a')$ and $p(x_1,x_1',y_1,y_1'|a,a')$ glued together. }
  \label{fig:0and1_gen}
\end{figure}

The case of non-perfectly hiding schemes is more involved. At first glance, one might expect that by proceeding analogously to the proof of
Theorem~\ref{thm:impossibility} --- i.e., gluing together $p(x_0,x_0',y_0,y'_0|a,a')$ and $p(x_1,x_1',y_1,y'_1|a,a')$ along $(x_0,x_0')$ and $(x_1,x_1')$ and defining $q$ the same way --- one can obtain a strategy $q$ that succeeds
with probability $1-\eps$ if the scheme is $\eps$-hiding. Unfortunately,
this approach fails because in order to show (NS1) we use that
$p(x_0,y_1|a,a') = p(x_1,y_1|a,a')$ which in general does not hold for
commitment schemes that are not perfectly hiding. 
As a consequence, our proof is more involved, and we have a constant-factor loss in the parameter.

\begin{theorem}
\label{thm:imperfect}
Let $\Com$ be a single-round two-prover commitment scheme
and suppose that it is $\eps$-hiding. Then there exists a non-signaling
two-prover strategy $q(x,x',y,y'|a,a',b,b')$ such that
$$
  \dProb _q [\Acc | 0] = \Prob _p [\Acc | 0] \quad\text{and}\quad
  \dProb _q [\Acc | 1] \geq \Prob _p [\Acc | 1] - 5\eps \, .
$$
Thus, if $\Com$ is perfectly sound, it is at best $(1-5\eps)$-binding.
\end{theorem}
To prove this result, we use two lemmas. In the first one, we add the
additional assumptions that $p(x_0|a,a') = p(x_1|a,a')$ and $p(x_0'|a,a')
= p(x_1'|a,a')$. The second one shows that we can tweak an arbitrary scheme
in such a way that these additional conditions hold. The proofs are given
in Appendix \ref{app:proof}.

\begin{lemma}
\label{lemma:imperfect1}
Let $\Com$ be a $\eps$-hiding two-prover commitment scheme with the
additional property that $p(x_0|a,a') = p(x_1|a,a')$ and
$p(x_0'|a,a') = p(x_1'|a,a')$. Then, there is a  non-signaling  $p'(x_1,x_1',y_1,y_1'|a,a')$ such that
$$d\bigl(p'(x_1,x_1',y_1,y_1'|a,a'),p(x_1,x_1',y_1,y_1'|a,a')\bigr)\leq\eps$$
and $p'(x_1,x_1'|a,a') = p(x_0,x_0'|a,a')$.
\end{lemma}
As usual, the non-signaling requirement on $p'(x_1,x_1',y_1,y_1'|a,a')$ is to be  understood as $p'(x_1,y_1|a,a') = p'(x_1,y_1|a)$ and $p'(x_1',y_1'|a,a') = p'(x_1',y_1'|a')$.

\begin{lemma}
\label{lemma:imperfect2}
Let $\Com$ be a $\eps$-hiding two-prover commitment scheme.
Then, there exists a non-signaling 
$\tilde{p}(x_1,x_1',y_1,y_1'|a,a')$ such that
$$d\bigl(\tilde{p}(x_1,x_1',y_1,y_1'|a,a'),p(x_1,x_1',y_1,y_1'|a,a')\bigr) \leq 2\eps$$
which has the property that 
$\tilde{p}(x_1|a,a') = p(x_0|a,a')$ and $\tilde{p}(x_1'|a,a') = p(x_0'|a,a')$.
\end{lemma}

With these two lemmas, Theorem~\ref{thm:imperfect} is easy to prove.

\begin{proof}[Proof of Theorem~\ref{thm:imperfect}]
We start with a $\eps$-hiding non-signaling bit-commitment scheme
$\Com$. We apply Lemma~\ref{lemma:imperfect2} and obtain a distribution
$\tilde p(x_1,x_1',y_1,y_1'|a,a')$ which is $2\eps$-close to 
$p(x_1,x_1',y_1,y_1'|a,a')$ and satisfies 
$\tilde{p}(x_1|a,a') = p(x_0|a,a')$ and $\tilde{p}(x_1'|a,a') = p(x_0'|a,a')$.
Furthermore, by triangle inequality
$$
d\bigl(\tilde{p}(x_1,x_1'|a,a'),p(x_0,x_0'|a,a')\bigr) \leq 3\eps \, .
$$
Thus, replacing $p(x_1,x_1',y_1,y_1|a,a')$ by
$\tilde{p}(x_1,x_1',y_1,y_1'|a,a')$ gives us a $3\eps$-hiding two-prover commitment scheme that satisfies the extra assumption in Lemma~\ref{lemma:imperfect1}. As a result, we obtain a distribution $p'(x_1,x_1',y_1,y_1'|a,a')$ that is $3\eps$-close to $\tilde{p}(x_1,x_1',y_1,y_1'|a,a')$, and thus $5\eps$-close to 
$p(x_1,x_1',y_1,y_1'|a,a')$, with the property that
$p'(x_1,x_1'|a,a') = p(x_0,x_0'|a,a')$. Therefore, replacing
$\tilde{p}(x_1,x_1',y_1,y_1'|a,a')$ by $p'(x_1,x_1',y_1,y_1'|a,a')$ gives us a {\em perfectly hiding} two-prover commitment scheme, to which we can apply Theorem~\ref{thm:impossibility}. As a consequence, there exists a non-signaling strategy $q(x,x',y,y'|a,a')$ with $\dProb _q[\Acc |0] = \Prob _p[\Acc |0]$
and $\dProb _q[\Acc |1] \geq \Prob _p[\Acc |1]-5\eps$, as claimed.
\end{proof}

\begin{remark}
If $\Com$ already satisfies $p(x_0|a,a') = p(x_1|a,a')$ and
$p(x_0'|a,a') = p(x_1'|a,a')$, we can apply Lemma~\ref{lemma:imperfect1} right
away and thus get a strategy $q$ with $\dProb _q[\Acc |0] = \Prob _p[\Acc |0]$
and $\dProb _q[\Acc |1]\geq \Prob _p[\Acc |1] - \eps$. Thus, with this
additional condition, we still obtain a tight bound as in
Theorem~\ref{thm:tight}.
\end{remark}

%=================================================================
\subsection{Multi-Round Schemes}
%=================================================================

We briefly discuss a limited extension of our impossibility results for single-round schemes to schemes where during the commit phase, there is multi-round interaction between the verifier $V$ and the two provers $P$ and $Q$. We still assume the opening phase to be one-round; this is without loss of generality in case of {\em classical} two-prover commitment schemes (where the honest provers are restricted to be classical). In this setting, we have the following impossibility result, which is restricted to perfectly-hiding schemes. 

\begin{theorem}
Let $\Com$ be a multi-round two-prover commitment scheme.
If $\Com$ is perfectly hiding, then there exists a non-signaling two-prover strategy that completely breaks the binding property, in the sense of Theorem~\ref{thm:impossibility}. 
\end{theorem}

A formal proof of this statement requires a definition of $n$-round non-signaling bipartite systems for arbitrary $n$. Such a definition can be based on the intuition that it must be possible to
query the left and right subsystem in any order. With this definition,
the proof is a straightforward extension of the proof of Theorem~\ref{thm:impossibility}: the non-signaling strategy is obtained by gluing together $p({\bf x}_0,{\bf x}_0'|{\bf a},{\bf a}')$ and $p({\bf x}_1,{\bf x}_1'|{\bf a},{\bf a}')$ along $({\bf x}_0,{\bf x}_0')$ and $({\bf x}_1,{\bf x}_1')$, and setting $q({\bf x},{\bf x}',y,y'|{\bf a},{\bf a}',b,b') := p({\bf x}_0,{\bf x}_0',y_b,y_{b'}'|{\bf a},{\bf a}')$, where we use bold-face notation for the vectors that collect the messages sent during the multi-round commit phase: ${\bf a}$ collects all the messages sent by the verifier to the prover $P$, etc. 

As far as we see, the proof of the non-perfect case, i.e. Theorem~\ref{thm:imperfect}, does not generalize immediately to the multi-round case. As such, proving the impossibility of {\em non-perfectly-hiding multi-round} two-prover commitment schemes remains an open problem.

%=================================================================
\section{Possibility of Three-Prover Commitments}
%=================================================================

It turns out that we can overcome the impossibility results by adding a third
prover. We will describe a scheme that is perfectly sound, perfectly hiding
and $2^{-n}$-binding with communication complexity $O(n)$.
We now define what it means for three provers to be non-signaling;
since our scheme is similar to a simple scheme, we can simplify this somewhat.
We consider distributions $q(x,y,z|a,b,c)$ where $a$ and $x$ are input
and output of the first prover $P$, $b$ and $y$ are input and output
of the second prover $Q$ and $c$ and $z$ are input and output of
the third prover $R$.
\begin{definition}
A conditional distribution $q(x,y,z|a,b,c)$ is called a {\em non-signaling (one-round)
tripartite system} if it satisfies 
\begin{align*}
  &q(x|a,b,c) = q(x|a) \;\text{,}\quad q(y|a,b,c) = q(y|b) \;\text{,}\quad q(z|a,b,c) = q(z|c) \;\text{,}\\
  &q(x,y|a,b,c) = q(x,y|a,b) \;\text{,}\quad q(x,z|a,b,c) = q(x,z|a,c)\\
  &\text{and}\quad q(y,z|a,b,c) = q(y,z|b,c) \, .
\end{align*}
\end{definition}
In other words, for any way of viewing $q$ as a bipartite system by dividing in- and outputs consistently into two groups, we get a non-signaling bipartite system. Actually, by means of Lemma~\ref{lemma:cut}, it is not hard to see that the first three requirements follow by the (union of the) latter three. 

We restrict to {\em simple} schemes, where during the commit phase, only $P$ is active, sending $x$ upon receiving $a$ from the verifier, and during the opening phase, only $Q$ and $R$ are active, sending $y$ and $z$ to the verifier, respectively. 

\begin{definition}
A {\em simple three-prover commitment scheme} $\Com$ consists of
a probability distribution $p(a)$, two distributions $p_0(x,y,z|a)$ and
$p_1(x,y,z|a)$, and an acceptance predicate $\Acc(x,y,z|a,b)$. \\
It is called classical/quantum/non-signaling if $p_b(x,y,z|a)$ is, when understood as a tripartite system $p_b(x,y,z|a,\emptyset,\emptyset)$ with two ``empty'' inputs. 
\end{definition}
Soundness and the hiding-property are defined in the obvious way. 
As for the binding property, for a simple three-prover commitment scheme $\Com$ and a non-signaling strategy $q(x,y,z|a,b,c)$, let
$$
  \Prob _q^* [\Acc |b]
  = \sum _{a,x,y,z} p(a)\cdot q(x,y,z|a,b,b)\cdot \Acc (x,y,z|a,b) \, .
$$
We say that $\Com$ is $\delta$-binding if
$$\Prob _q^* [\Acc |0] + \Prob _q^*[\Acc |1]
  \leq 1 +\delta\text .$$

\begin{theorem}
\label{thm:multi}
For every positive integer $n$, there exists a classical simple three-prover commitment
scheme that is perfectly sound, perfectly hiding and $2^{-n}$-binding. The
verifier communicates $n$ bits to the first prover and receives $n$ bits from
each prover.
\end{theorem}
The scheme that achieves this is essentially the same as the example two-prover scheme described in the introduction, except that we add a third prover that imitates the actions of the second. 
To be more precise: the provers $P$, $Q$ and $R$ have as shared
randomness a uniformly random $r \in \set{0,1}^n$.
The verifier $V$ chooses a uniformly random $a \in \set{0,1}^n$ and sends
it to $P$. As commitment, $P$ returns $x := r \oplus a\cdot b$. To open the commitment to $b$, $Q$ and $R$ send $y:= r$ and $z := r$ to $V$ who accepts
if and only if $y = z$ and $x = y \oplus a\cdot b$. 

Before beginning with the formal proof that this scheme has the properties
stated in our theorem, we give some intuition.
Let $a$ and $x$ be the input and output of the dishonest first prover, $P$.
To succeed, the second prover $Q$ has to produce output $x\oplus a\cdot b$
where $b$ is the second prover's input and the third prover $R$ has
to produce $x\oplus a\cdot c$ where $c$ is the third prover's input.
Our theorem implies that a strategy which always produces these outputs must
be signaling. Why is that the case?

In the game that defines the binding-property, we always have $b=c$,
but the dishonest provers must obey the non-signaling constraint even in the
``impossible'' case that $b\neq c$. Let us consider the XOR of $Q$'s output
and $R$'s output in the case that $b\neq c$:
we get $(x \oplus a\cdot b) \oplus (x \oplus a\cdot c) =
a\cdot b \oplus a\cdot c = a$. But in the non-signaling setting, the joint
distribution of $Q$'s and $R$'s output may not depend on $a$. Thus,
the strategy we suggested does not satisfy the non-signaling constraint.
Let us now prove the theorem.
\begin{proof}[Proof of Theorem \ref{thm:multi}]
It is easy to see that the scheme is sound. Furthermore, for every fixed
$a$ and $b$, $p_b(x|a)$ is uniform, so the scheme is perfectly
hiding. Now consider a non-signaling strategy $q$ for dishonest provers.
The provers succeed if and only if $y = z = x \oplus a\cdot b$.
Define $q(a,x,y,z|b,c) = p(a)\cdot q(x,y,z|a,b,c)$.
The non-signaling property implies that
\begin{align}
  q(y = x\oplus a\cdot b|a,b,c=0) &= q(y = x\oplus a\cdot b|a,b,c=1) \quad\text{and}
  \label{eq:nonsig1}\\
  q(z = x\oplus a\cdot c | a,b=0,c) &= q(z=x\oplus a\cdot c|a,b=1,c) \, .
  \label{eq:nonsig2}
\end{align}
It follows that
\begin{align*}
\dProb _q[\Acc |0]& + \dProb _q[\Acc |1]\\
=\ &q(y = x\oplus a\cdot b, z = x\oplus a\cdot c|b=0, c=0)\\
&\quad + q(y = x\oplus a\cdot b, z = x\oplus a\cdot c|b=1, c=1)\\
\leq\ &q(y = x\oplus a\cdot b|  b=0, c= 0) + q(z=x\oplus a\cdot c|b=1,c=1)\\
=\ &q(y = x\oplus a\cdot b|b=0,c=1) + q(z=x\oplus a\cdot c|b=0,c=1)\\
  &\quad \text{by Equations \eqref{eq:nonsig1} and \eqref{eq:nonsig2}}\\
\leq\ &1+q(y = x\oplus a\cdot b, z=x\oplus a\cdot c|b=0,c=1)
  \text{ by Equation \eqref{eq:pr_sum}}
\end{align*}
It now remains to upper-bound
$q(y = x\oplus a\cdot b, z=x\oplus a\cdot c|b=0,c=1)$. Since $p(a)$ is
uniform and $q(y,z|a,b,c)$ is independent of $a$, we have
$$
  q(y=x\oplus a\cdot b, z=x\oplus a\cdot c|b=0,c=1)
  \leq q(y\oplus z = a|b=0,c=1)
  = \frac{1}{2^n}
$$
and thus our scheme is $2^{-n}$-binding.
\end{proof}

\begin{remark}\label{rem:issue}
The three-prover scheme above has the drawback that {\em two} provers are involved in the opening phase; as such, there needs to be {\em agreement} on whether to open the commitment or not; if there is disagreement then this may be problematic in certain applications. However, $P$ and $Q$ are not allowed to communicate. One possible solution is to have $V$ forward an {\em authenticated} ``open'' or ``not open'' message from $P$ to $Q$ and $R$. This allows for some communication from $P$ to $Q$ and $R$, but if the size of the authentication tag is small enough compared to the security parameter of the scheme, i.e., $n$, then security is still ensured. 
\end{remark}

%=================================================================
\subsection*{Acknowledgements}
%=================================================================

We would like to thank Claude Cr{\'e}peau for pointing out the issue addressed in Remark~\ref{rem:issue} and the solution sketched there, and Jed Kaniewski for helpful discussions regarding relativistic commitments.

\bibliographystyle{alpha}
%\bibliographystyle{plain}
% %\bibliographystyle{abbrv}
\bibliography{FF}
% \bibliography{qip,crypto,procs}

\begin{appendix}

\section{Capturing the Non-signaling property}
\label{app:sufficient}
In this section, we argue that Definition \ref{def:ns2} is not only necessary
but also sufficient to capture the non-signaling constraint.
Consider a two-round bipartite system that conforms to Definition \ref{def:ns2}.
We show that the two subsystems can be queried {\em in any order} without
altering the output distribution, as long as the order of rounds for each
subsystem individually is respected. Thus, it is impossible to obtain
information about the right side of the system by observing only the behaviour
on the left side (and vice versa), which shows that Definition \ref{def:ns2} is
indeed sufficient. First, we point out the following. 

\begin{remark}\label{rem:connections}
(C1) and (NS1) together imply that $q(x|a,b)$ and $q(x|a,a')$ are well-defined and satisfy
$$
q(x|a,b) = q(x|a) \quad\text{(C3)}  \qquad\text{and}\qquad q(x|a,a') = q(x|a) \qquad \text{(NS3)} \, .
$$
This follows from Lemma~\ref{lemma:cut} below. 
\end{remark}

\begin{lemma}\label{lemma:cut}
Any conditional distribution $q(x|a,b,c,d)$ that satisfies $q(x|a,b,c,d) = q(x|a,b)$ as well as $q(x|a,b,c,d) = q(x|a,c)$, must also satisfy $q(x|a,b,c,d) = q(x|a)$. 
\end{lemma}

\begin{proof}
Recall that, by convention, $q(x|a,b,c,d) = q(x|a,b)$ means $q(x|a,b,c,d) = q(x|a,b,c',d')$ for all $x,a,b,c,c',d,d'$, and similarly for $q(x|a,b,c,d) = q(x|a,c)$. As such, for arbitrary $x,a,b,b',c,c',d,d'$ it holds that
$$q(x|a,b,c,d) = q(x|a,b,c',d') = q(x|a,b',c',d')$$
and thus $q(x|a,b,c,d) = q(x|a)$. 
\end{proof}
If $q(x,x',y,y'|a,a',b,b')$ is a non-signaling two-round bipartite system, then it can be written as
\begin{align*}
q(x,x',y,y'|a,a',b,b') &= q(x,y|a,b) \cdot q(x',y'|x,y,a,a',b,b') \\
&= q(x|a) \cdot q(y|x,a,b) \cdot q(x'|x,y,a,a',b) \cdot q(y'|x,y,a,a',b,b')
\end{align*}
where the first equality uses (NS1), and the second uses (C3) and (C2), and as
\begin{align*}
&q(x,x',y,y'|a,a',b,b')\\
=\quad &q(x,x'|a,a') \cdot q(y,y'|x,x',a,a',b,b')\\
=\quad & q(x|a) \cdot q(x'|x,a,a') \cdot q(y|x,x',a,a',b) \cdot q(y'|x,x',y,a,a',b,b')
\end{align*}
where the first equality uses (C1), and the second uses (NS3) and (NS2), and the second equality can also be replaced by 
$$
= q(x|a) \cdot q(x'|x,a,a') \cdot q(y'|x,x',a,a',b') \cdot q(y|x,x',y,a,a',b,b') \, .
$$
And, similarly, with the roles of the primed and unprimed variables exchanged. 
This shows that the two subsystems can be queried in any order. For instance, one can first query the left subsystem to get $x$ on input $a$, distributed according to $q(x|a)$, and then $y$ on input $b$, distributed according to $q(y|x,a,b)$, and then then one can query the right subsystem twice to get $x'$ and $y'$, distributed according to $q(x'|x,y,a,a',b)$ and $q(y'|x,y,a,a',b,b')$, respectively.%
\footnote{Note that in oder to sample, say, $x'$ according to $q(x'|x,y,a,a',b)$, it seems like that the right subsystem needs to know $a,x$ etc., i.e., that communication is necessary, contradicting the non-signaling requirement. However, this reasoning merely shows that in general, such a non-signaling system is not classical. 
 }
Or, one can first query the left subsystem once to obtain $x$, then query the right subsystem to obtain $x'$ etc.  
It is straightforward to verify that all six eligible orderings are possible. 
\section{Proofs of Lemma \ref{lemma:imperfect1} and Lemma
  \ref{lemma:imperfect2}}
\label{app:proof}
\begin{proof}[Proof of Lemma \ref{lemma:imperfect1}]
For arbitrary $a$ and $a'$, we use Lemma~\ref{lemma:gluing} to glue together the distributions $p(x_0,x_0',y_0,y_0'|a,a')$ and $p(x_1,x_1',y_1,y_1'|a,a')$ to obtain a joint distribution 
$p(x_0,x_0',x_1,x_1',y_0,y_0',y_1,y_1'|a,a')$ such that
$$
p\bigl((x_0,x_0')\neq (x_1,x_1')|a,a'\bigr)\leq \eps \, ,
$$
and thus $d(p(x_0,x_0',y_1,y_1'|a,a'),p(x_1,x_1',y_1,y_1'|a,a'))\leq \eps$.
Let $\Lambda$ be the event that both $x_0 = x_1$ and $x_0' = x_1'$. We define $p'(x_1,x_1',y_1,y_1'|a,a')$ as follows, where $x_0$ is associated with $x_1$ and $x'_0$ with $x'_1$:
\begin{align*}
  p'(x_1,x_1',y_1,y_1'|a,a')
  &:= p(\Lambda,x_0,x_0'|a,a')\cdot p(y_1,y_1'|\Lambda,x_1,x_1',a,a')\\
  &\qquad + p(\overline{\Lambda},x_0,x_0'|a,a')\cdot r(y_1|x_0,a,a')
    \cdot r(y_1'|x_0',a,a')\\
  &= p(\Lambda,x_1,x_1',y_1,y_1'|a,a')\\
  & \qquad +p(\overline{\Lambda},x_0,x_0'|a,a')
    \cdot r(y_1|x_0,a,a')\cdot r(y_1'|x_0',a,a')
\end{align*}
where $r(y_1|x_0,a,a')$ and $r(y_1'|x_0',a,a')$ are to be defined later, and the last equality holds by definition of~$\Lambda$.%
\footnote{Algorithmically, the distribution $p'$ should be understood as follows. First, $x_0,x_0',x_1$ and $x_1'$ are sampled according to the glued-together distribution $p$. Then, if the event $\Lambda$ occurred (i.e.~$x_0 = x_1$ and $x'_0 = x'_1$), $y_1$ and $y'_1$ are sampled according to the corresponding conditional distribution; otherwise, they are chosen {\em independently} according to distributions that depend only on $x_0$ and $x'_0$, respectively. }

The claim about the closeness to $p(x_1,x_1',y_1,y_1'|a,a')$ follows from the fact that $p(\overline{\Lambda}|a,a') \leq \eps$. Furthermore, we have $p'(x_1,x_1'|a,a')
= p(\Lambda,x_0,x_0'|a,a') + p(\overline{\Lambda},x_0,x_0'|a,a') = p(x_0,x_0'|a,a')$ as claimed.

It remains to show that we can achieve $p'$ to be non-signaling. 
For that, we simply define $r(y_1|x_0,a,a')$, and similarly $r(y_1'|x_0',a,a')$, in such a way that $p'(x_1,y_1|a,a') = p(x_1,y_1|a,a')$; this does the job since $p(x_1,y_1|a,a') = p(x_1,y_1|a)$, and as such $p'(x_1,y_1|a,a') = p'(x_1,y_1|a)$. 
Note that
\begin{equation}
  \label{eq:def_r}
  p'(x_1,y_1|a,a') = p(\Lambda,x_1,y_1|a,a')
  + p(\overline{\Lambda},x_0|a,a')\cdot r(y_1|x_0,a,a') \, .
\end{equation}
Thus, we set
$$
  r(y_1|x_0,a,a') := \frac{p(x_1,y_1|a,a')-p(\Lambda,x_1,y_1|a,a')}
    {p(\overline{\Lambda},x_0|a,a')}
  = \frac{p(\overline{\Lambda},x_1,y_1|a,a')}{p(\overline{\Lambda},x_0|a,a')}
$$
It remains to show that $r(y_1|x_0,a,a')$ as defined is indeed a probability distribution, and that things work out also in case $p(\overline{\Lambda},x_0|a,a') = 0$. 

In the latter case, we have $p'(x_1,y_1|a,a') = p(\Lambda,x_1,y_1|a,a')$, independent of the choice of $r$; thus, it remains to show that $p(\Lambda,x_1,y_1|a,a') = p(x_1,y_1|a,a')$.  For that, we observe that $p(\Lambda,x_1|a,a') = p(\Lambda,x_0|a,a') = p(x_0|a,a') = p(x_1|a,a')$, where the first equality is due to the definition of $\Lambda$ and the last holds
by our additional assumption on $\Com$. It follows that
$$
  \sum _{y_1}p(\Lambda,x_1,y_1|a,a') = p(\Lambda,x_1|a,a') = p(x_1|a,a')
  = \sum _{y_1}p(x_1,y_1|a,a')
$$
and since $p(\Lambda,x_1,y_1|a,a') \leq p(x_1,y_1|a,a')$, it holds that $p(\Lambda,x_1,y_1|a,a') = p(x_1,y_1|a,a')$ as required.

Finally, to show that $r(y_1|x_0,a,a')$ is a probability distribution, we observe that $r(y_1|x_0,a,a') \geq 0$, and, summing over
$y_1$ and using that $p(x_0|a,a') = p(x_1|a,a')$, we see that
\begin{align*}
  \sum _{y_1} r(y_1|x_0,a,a')
  = \frac{p(x_1|a,a') - p(\Lambda,x_1|a,a')}{p(\overline{\Lambda}, x_0|a,a')}
  &= \frac{p(x_0|a,a') - p(\Lambda,x_0|a,a')}{p(\overline{\Lambda}, x_0|a,a')}\\
  &= \frac{p(\overline{\Lambda}, x_0|a,a')}{p(\overline{\Lambda}, x_0|a,a')}\\
  &= 1 \, .
\end{align*}
In the same way, it is possible to choose $r(y_1'|x_0',a,a')$ so that $p'(x'_1,y'_1|a,a') = p(x'_1,y'_1|a,a') = p(x'_1,y'_1|a')$, using the assumption that
$p(x_0'|a,a') = p(x_1'|a,a')$. This concludes the proof.
\end{proof}
\begin{proof}[Proof of Lemma \ref{lemma:imperfect2}]
We begin by adjusting the distribution of $x_1$. By the hiding property of $\Com$,  $p(x_0,x_0'|a,a')$ and $p(x_1,x_1'|a,a')$ are $\eps$-close, 
and thus in particular $d(p(x_0|a,a'), p(x_1|a,a'))\leq \eps$.
Gluing together the distributions $p(x_0|a,a')$ and $p(x_1,x_1',y_1,y_1'|a,a')$
along $x_0$ and $x_1$, we get $p(x_0,x_1,x_1',y_1,y_1'|a,a')$ such that 
$$
p'(x_1,x_1',y_1,y_1'|a,a') := p(x_0,x_1',y_1,y_1'|a,a')
$$ 
satisfies  $d\bigl(p'(x_1,x_1',y_1,y_1'|a,a'),p(x_1,x_1',y_1,y_1'|a,a')\bigr)\leq\eps$ and also $p'(x_1|a,a') = p(x_0|a,a')$.

We show that $p'$ is non-signaling. Since $p'(x_1',y_1'|a,a') = p(x_1',y_1'|a,a')$ and $p$ is non-signaling,
it follows that $p'(x_1',y_1'|a,a') = p'(x_1',y_1'|a')$. 
Showing that $p'(x_1,y_1|a,a') = p'(x_1,y_1|a)$ is equivalent to showing that $p(x_0,y_1|a,a') = p(x_0,y_1|a)$. 
By the observation in Remark~\ref{rem:commute}, the marginal $p(x_0,x_1,y_1|a,a')$ is obtained by gluing together $p(x_0|a,a')$ and $p(x_1,y_1|a,a')$ along $x_0$ and $x_1$.
Since $\Com$ is non-signaling, it holds that $p(x_0|a,a') = p(x_0|a)$ and
$p(x_1,y_1|a,a') = p(x_1,y_1|a)$. It follows that
$p(x_0,x_1,y_1|a,a') = p(x_0,x_1,y_1|a)$, and therefore that
$p(x_0,y_1|a,a') = p(x_0,y_1|a)$.

% Looking at the construction in the proof of Lemma~\ref{lemma:gluing},
% we see that
% $$
%   p(x_0,y_1) = \sum _{x_1}p(x_0,x_1|a,a')\cdot p(y_1|x_1,a,a')
% $$
% Since $p$ is non-signaling, we have $p(x_1,y_1|a,a') = p(x_1,y_1|a)$ and
% thus, $p(y_1|x_1,a,a') = p(y_1|x_1,a)$. Since $p(x_0,x_1|a,a')$ is obtained
% by applying Proposition~\ref{prop:gluing} to $p(x_0|a,a') = p(x_0|a)$ and
% $p(x_1|a,a') = p(x_1|a)$, it also holds that $p(x_0,x_1|a,a') = p(x_0,x_1|a)$.
% This shows that $p'(x_1,y_1|a,a')$ is indeed independent of $a'$. 

In order to obtain $\tilde{p}$ as claimed, we repeat the above process. Note that the modification from $p$ to $p'$ did not change the distribution of $x_1',y_1'$, i.e., $p'(x_1',y_1'|a,a') = p(x_1',y_1'|a,a')$, and thus in particular $d\bigl(p(x'_0|a,a'), p'(x'_1|a,a')\bigr) = d\bigl(p(x'_0|a,a'), p(x'_1|a,a')\bigr) \leq \eps$. Therefore, exactly as above, we can now adjust the distribution of $x'_1$ in $p'$ and obtain a non-signaling $\tilde{p}(x_1,x_1',y_1,y_1'|a,a')$ that is $\eps$-close to $p'(x_1,x_1',y_1,y_1'|a,a')$ and thus $2\eps$-close to $p(x_1,x_1',y_1,y_1'|a,a')$, and which satisfies $\tilde{p}(x_1'|a,a') = p(x_0'|a,a')$ and 
$\tilde{p}(x_1|a,a') = p'(x_1|a,a') = p(x_0|a,a')$, as claimed. 
% We repeat the gluing-procedure for $x_0',x_1'$: since
% $$d(p(x_0'|a,a'),p'(x_1'|a,a')) = d(p(x_0'|a,a'),p(x_1'|a,a')) \leq \eps$$
% by gluing along $x_0',x_1'$, we can define
% $p'(x_0,x_0',x_1,x_1',y_0,y_0',y_1,y_1'|a,a')$ such that
% $$
%   d(p'(x_1,x_0',y_1,y_1'|a,a'), p'(x_1,x_1',y_1,y_1'|a,a'))\leq \eps
% $$
% Now we can define the distribution that fulfils our claim:
% $$
%   \tilde{p}(x_1,x_1',y_1,y_1'|a,a') = p'(x_1,x_0',y_1,y_1'|a,a')
% $$
% By construction of $\tilde{p}$ it follows that
% $\tilde{p}(x_1'|a,a') = p'(x_0'|a,a') = p(x_0'|a,a')$,
% $\tilde{p}(x_1|a,a') = p'(x_1|a,a') = p(x_0|a,a')$ and
% $d(\tilde{p}(x_1,x_1',y_1,y_1'|a,a'),p'(x_1,x_1',y_1,y_1'|a,a')\leq\eps$
% so
% $$d(\tilde{p}(x_1,x_1',y_1,y_1'|a,a'),p(x_1,x_1',y_1,y_1'|a,a')) \leq 2\eps$$
% Arguing as before, one can show that $\tilde{p}$ is non-signaling and thus
% our claim holds.
\end{proof}

\end{appendix}

\end{document}